\documentclass[AMA,STIX1COL]{WileyNJD-v2}

\articletype{Article Type}%

\received{26 April 2016}
\revised{6 June 2016}
\accepted{6 June 2016}

\begin{document}

\title{Unpredictable solutions of Duffing type equations with Markov coefficients}

\author[1]{Marat Akhmet}

\author[2,3]{Madina Tleubergenova}

\author[1,2]{Akylbek Zhamanshin}

\authormark{MARAT AKHMET \textsc{et al}}

\address[1]{\orgdiv{Department of Mathematics}, \orgname{Middle East Technical University}, \orgaddress{\state{Ankara}, \country{Turkey}}}

\address[2]{\orgdiv{Department of Mathematics}, \orgname{Aktobe Regional University}, \orgaddress{\state{Aktobe}, \country{Kazakhstan}}}

\address[3]{\orgname{Institute of Information and Computational Technologies}, \orgaddress{\state{Almaty}, \country{Kazakhstan}}}

\corres{Marat Akhmet, \email{marat@metu.edu.tr}}


\abstract[Summary]{The paper considers a stochastic differential equation of Duffing type with Markov coefficients. The existence of unpredictable solutions is considered. The unpredictability is a property of bounded functions characterized by  unbounded sequences of moments of divergence and convergence in Bebutov dynamics. Markov components of the equation coefficients admit the unpredictability property. The components of the equation coefficients are derived from  a Markov chain. The existence, uniqueness and exponential stability of an unpredictable solution are proved. The sequences of divergence and convergence of the coefficients and the solution are synchronized. Numerical example that support the theoretical results are provided. 
}



\maketitle


\section{Introduction}\label{sec1}

It is of great importance to study  a Duffing equation with variable coefficients. It was emphasized in the book by Moon \cite{Moon2004} that  the case when the coefficients are irregular is of strong interest. The question is, what if the perturbations are random, and are of noise type, and the noise is articulated with asymptotic properties of divergence and convergence. Obviously, there are two problems appear. The first one is, how to insert the deviations in coefficients and to be able for proper evaluations. The second task is, how  the stochastic processes relate to  what we understand as deterministic chaos. That is, not to say that a random process is a deterministic phenomenon, but to demonstrate that some significant features recognized for the chaos can be seen in dynamics originated, for example, from Markov  events, which happen with probabilities. This two questions are concerned in the present article. The most related  results which are utilized to answer the questions, have been already obtained in our previous research. In papers \cite{AkhmetFen2016,Akhmet2018}, we introduced a new type of recurrence, the unpredictable point, which is an unpredictable function in Bebutov dynamics. In the research of article \cite{Akhmet2022}, it was proved that any infinite time realization of a Markov process  with finite state space  and without memory is an unpredictable sequence. 

 Most common form of stochastic differential equation (SDE)   is a differential equation with one or more stochastic processes  as terms.  The solutions of SDEs are also stochastic processes. Typically, a SDE  is an ordinary differential equation perturbed by a  term, which depends on a white noise variable calculated as the derivative of Brownian motion or the Wiener process.  SDEs  often are understood as continuous time limit of stochastic difference equations. 

In our research, we follow  the suggestion to consider  realizations of random dynamics as functions, in general,  and sequences, in particular.  As it is said in the book     \cite{Nicolis1989} 
"We have described stochastic dynamics in terms of probability distributions and their various moments. A complimentary, and for many purposes especially illuminating approach, is the study of individual outcomes of the stochastic process of interest."  We agree with the authors, and think that the outcomes have to be considered  not only for applications, but also as perturbations for various theoretical models.   In the present paper,  both:  inputs, which are  the cofficients of the equation of the Duffing type, and correspondingly,  outputs,  that is, the solutions of the equation are individualized.   

In the present study, stochastic processes appear in various roles.  The first, it is the discrete Markov chain with finite state space and without memory (one can use processes with memory in future extension of the method). A realization of the chain is  applied as an input for a dissipative  stochastic inhomogeneous equations. And the random solutions of the equations in their own turn are used as coefficients and inputs for the stochastic Duffing equation.
Finally, it is approved that the solution of the equation of Duffing type is  a continuous unpredictable function. It is clear that the scheme of the present study can be extended for other many theoretical  tasks as well as  applications.

The concept of unpredictability was introduced in papers \cite{AkhmetFen2016,Akhmet2018}, and has been applied for various problems of differential equations, neural networks and in gas discharge-semiconductor systems \cite{AkhmetN2020,AkhmetZ2020,AkhmetM2020,AkhmetK2022}. It is powerful instrument for chaos indication \cite{Akhmetbook2021,AkhmetA2021,AkhmetA2019,AkhmetM2022,Miller2019,Thakur2020}. 

The Markov research \cite{Markov1971} was  considered to show that  random processes of dependent events can also behave as independent events. Thus, simple dynamics  were invented, which have been approved as  most effective for many applications. It is impossible underestimate the role of the Markov processes  in development of random dynamics theory and its applications. For example, the ergodic theorem was strictly approved at the first time for  the dynamics. There are several  observations that the chains are strongly connected to  symbolic dynamics and to Bernoulli scheme. The final step for the comprehension was done by Donald Ornstein, who verified that $B$-automorphisms such as subshifts of finite type and Markov shifts, Anosov flows and Sinai’s billiards, ergodic automorphisms of the $n$-torus and the continued fraction transforms are, in fact, isomorphic   \cite{Ornstein1970}.  Considering these results  it is of great  necessity to show that various random processes can be described in terms of chaos, and that they relate equally in the sense. Investigators have worked   in both directions, for chaos in random dynamics, as well as for stochastic features in deterministic motions \cite{Bowen1970}. Thus, the problem of chaos in Markov chains, which is in  focus of our interest, is a part of the more general  and significant project.

The unpredictable orbit \cite{Akhmet2016} as a single isolated motion, presenting the Poincar\'{e} chaos \cite{AkhmetFen2016}, was identified as a certain event in the Markov chains  \cite{Akhmet2022},  and our present  results are not surprising in this sense, if one issues from the research in \cite{Ornstein1970} and  \cite{Akhmet2022}.
 
The Duffing equation has the form \cite{Duffing1918}
\begin{eqnarray}\label{OriginalDuffing}
	x''+ax'+bx+cx^3=F_0cos(\lambda t),
\end{eqnarray}
where $a$ is the damping coefficient, $b$ and $c$ are stiffness (restoring) coefficients, $F_0$ is the coefficient of excitation, $\lambda$ is the frequency of excitation and $t$ is the time. The major part of papers on the equation assume that  the coefficients  $a,b,c$ and $F_0$ are constant \cite{Liu2015,Zeng1997}. Considering  the original  model one can  assume   mechanical reasons for variable  coefficients. For example,  not constant  damping  and  driving  force \cite{Estevez2011}.

The main subject of this article is the following stochastic differential equation (SDE)
\begin{eqnarray} \label{duffing}
	x''(t)+(p_0(t)+p_1(t))x'(t)+(q_0(t)+q_1(t))x(t)+(r_0(t)+r_1(t))x^3(t)=(F_0(t)+F_1(t))cos(\lambda t),
\end{eqnarray}
where $t, x\in \mathbb{R};$ $\lambda$ is a real constant;  $p_0(t),$ $q_0(t),$ $r_0(t)$ and $F_0(t)$ are continuous periodic functions; coefficient components $p_1(t),$ $q_1(t),$ $r_1(t)$ and $F_1(t)$ are derived from realizations of Markov processes. This is why, we say that the coefficients are Markovian.  Let us remind that the right-hand-side of the equation is also assumed as a coefficient.  If the periodic components of the coefficients  are inserted  for the stability of the  solution of equation (\ref{duffing}), whereas Markov components cause irregularity of solutions.  It is important to emphasize that the main goal of the research not to approve a chaos for the output. But to show existence of the stochastic output for the equation, which admits the unpredictability property, and moreover, the property of the output is synchronized with the  asymptotic characteristics  of stochastic perturbations in the model.

\section{Preliminaries}\label{sec2}

In this section the definitions of unpredictable and Poisson stable functions as well as definition of unpredictable sequences are given. Moreover, the algorithm of construction for Markovian coefficients of SDE (\ref{duffing}) is provided.

\subsection{Unpredictable functions}

The following definitions are basic in the theory of unpredictable points, orbits and functions introduced \cite{AkhmetFen2016,Akhmet2018} and developed further in papers \cite{AkhmetN2020,AkhmetZ2020,AkhmetM2020,AkhmetK2022,Akhmetbook2021,AkhmetA2021,AkhmetA2019,AkhmetM2022,Miller2019,Thakur2020}.

\begin{definition} \cite{Akhmet2018} \label{def1}
	A uniformly continuous and bounded function $\psi : \mathbb  R  \rightarrow \mathbb  R$ is unpredictable if there exist positive numbers $\epsilon_{0}, \sigma$ and sequences $t_{n}, s_{n}$ both of which diverge to infinity such that $|\psi(t+t_{n})-\psi(t)|\rightarrow 0$ as $n\rightarrow\infty$ uniformly on compact subsets of $\mathbb  R$ and $|\psi(t+t_{n})-\psi(t)|\geq\epsilon_{0}$ for each $t\in [s_{n}-\sigma, s_{n}+\sigma]$ and $ n\in \mathbb  N$.
\end{definition}

In what follows, we shall call $t_n$ and $s_n$ as convergence and divergence sequences, respectively. The presence of the convergence sequence is the argument that any unpredictable function is Poisson stable \cite{Akhmet2022,Akhmet2016,Sell1971}, but not vice versa.

\begin{definition}\cite{Sell1971} \label{Poisson1} A  function $\phi(t)$$:$ $\mathbb{R}\rightarrow\mathbb{R},$ bounded and continuous,  is said to be Poisson stable if there is a sequence of moments $t_n,$ $t_n\rightarrow \infty$ as $n \rightarrow \infty,$ such that the sequence  $\phi(t+t_n)$ uniformly converges to $\phi(t)$ on each bounded interval of the real axis.
\end{definition}

The discrete version of the Definition \ref{def1} is as follows.
\begin{definition}\cite{Akhmet2018}
	A bounded sequence $\{k_i\} \in \mathbb R,$ $i \in \mathbb{Z},$  is called unpredictable if there exist a positive number $\epsilon_0$ and the sequences $\{\zeta_n\}, \{\eta_n\},$ $n \in \mathbb{N},$ of positive integers both of which diverge to infinity such that $|k_{i+\eta_n}-k_i| \rightarrow 0$ as $n \rightarrow \infty$ for each $i$ in bounded intervals of integers and $|k_{\zeta_n+\eta_n}-k_{\eta_n}|\geq \epsilon_0$ for each $n \in \mathbb{N}.$
\end{definition}

In this paper, we shall consider unpredictable sequences with non-negative arguments and call them also unpredictable sequences \cite{Akhmetbook2021}.

Let us give examples of unpredictable functions. Using an unpredictable sequence, $k_i,$ one can construct a piecewise constant function $\phi(t)$, such that $\phi(t)=k_i$ on intervals $t\in [hi,h(i+1)),$ where $h$ is a real number. In papers \cite{AkhmetN2020,AkhmetZ2020}, the function $\phi(t)$ is determined through the solution of the logistic map and the Bernoulli process is used. Another unpredictable function, $W(t),$ is a continuous solution of differential equation $W'(t)=\alpha W(t)+\phi(t),$ where $\alpha$ is a negative number. In Figure \ref{duf} (a) the graph of function $\phi(t)=\lambda_i,$ for $t \in [i,i+1),$ $i=0,1,2,...,$ is shown, where $\lambda_i$ is the unpredictable solution \cite{AkhmetFen2016} of the logistic map, $\lambda_{i+1}=\mu \lambda_i(1-\lambda_i),$ $i \in \mathbb Z,$ with $\lambda_0=0.4,$ $\mu=3.9.$ Figure \ref{duf} (b) depicts the graph of the solution, $w(t),$ of the equation with $w(0)=0.6$ and $\alpha=-2,$ which exponentially approaches to the unique unpredictable solution, $W(t),$ of the non-homogeneous equation. This is why, the red line can be considered \cite{AkhmetM2020} for t>40 as the graph of an unpredictable function. In the present paper, the coefficients of the SDE (\ref{duffing}) are determined by applying the algorithm for $W(t),$ but randomly such that a Markov chain is used instead of the logistic equation. 

\begin{figure}[h]
	\begin{minipage}[h]{0.49\linewidth}
		\center{\includegraphics[width=0.8\linewidth]{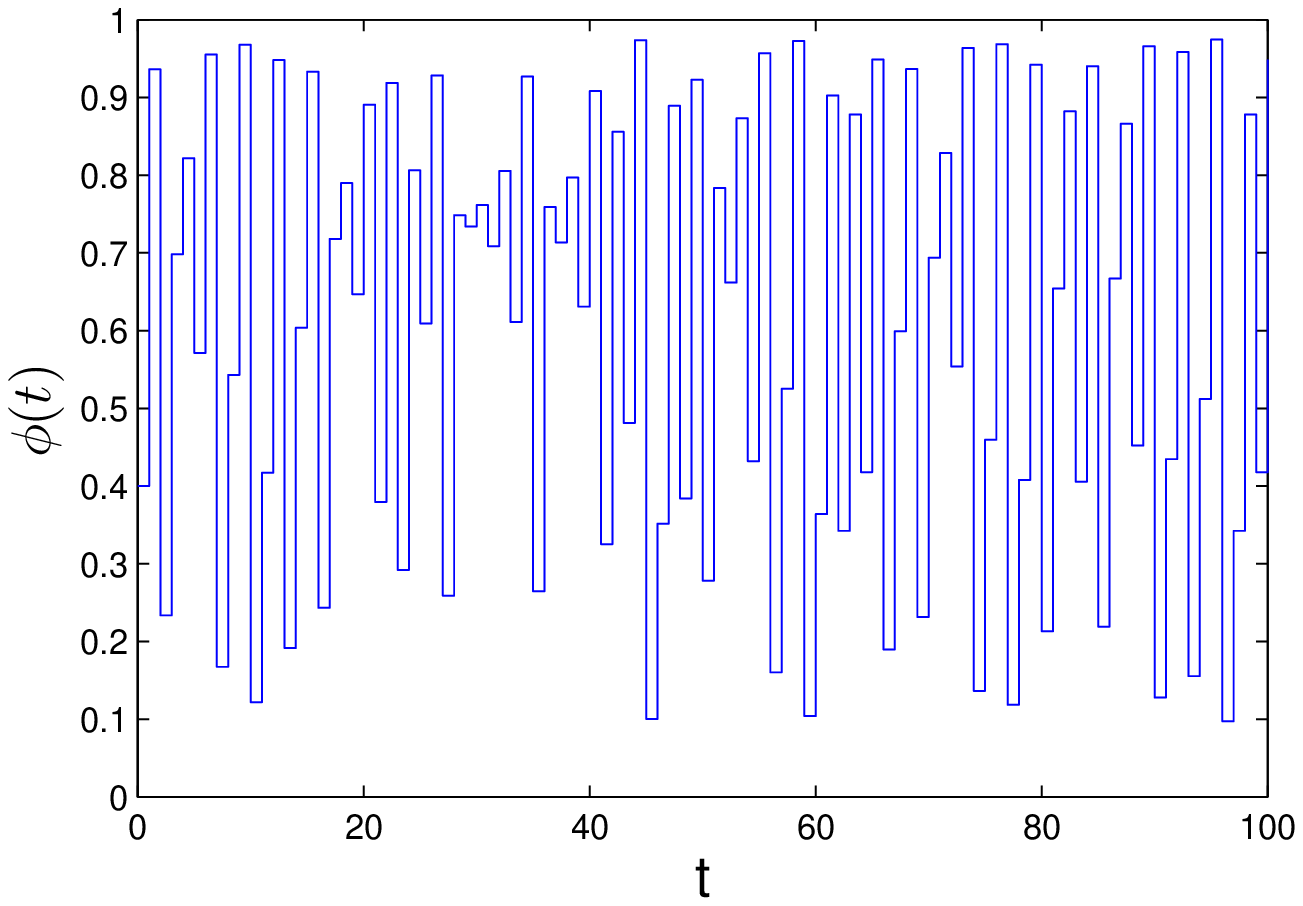} \\ a)}
	\end{minipage}
	\hfill
	\begin{minipage}[h]{0.49\linewidth}
		\center{\includegraphics[width=0.8\linewidth]{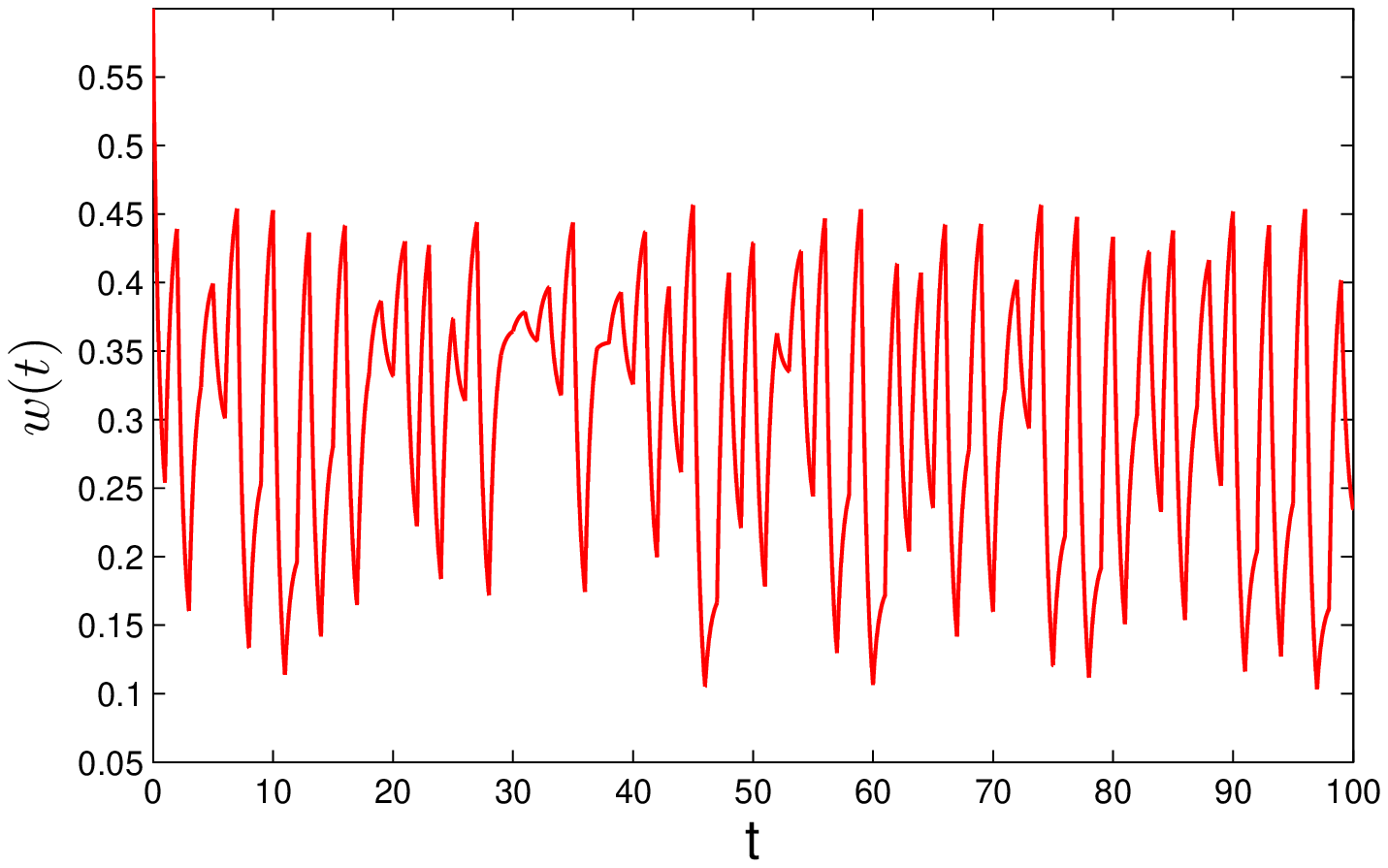} \\ b)}
	\end{minipage}
	\caption{The graphs of the discontinuous and continuous functions, $\phi(t)$ and $w(t).$}
	\label{duf}
\end{figure}


\subsection{Markovian coefficients} \label{markfunction}

In this part of the paper, we demonstrate algorithms how to construct Markovian coefficients for Duffing type equation (\ref{duffing}).

 A Markov chain is a stochastic model,   which   describes a sequence of possible events   such  that   the probability of each event depends only on the state attained in the previous   one \cite{Hajek2015,Karlin2012,Meyn2009}.	

Since  we expect    for  the  chaotic dynamics realizations to be  bounded,   the special Markov chain with  boundaries    is  constructed below. Let  the    real  valued scalar  dynamics 
\begin{eqnarray}\label{Markovchainn}
X_{n+1}= X_n + Y_n,  n  \ge 0,
\end{eqnarray}
be given  such  that   $Y_n$ is a random variable with values in $\{-2,2\}$ with probability  distribution $P(2) =P(-2) = 1/2,$  if $X_n \neq  -4,6,$   and   certain events $Y_n = 2,$  if $X_n =-4,$   and $Y_n = -2,$ if $X_n = 6.$     To    satisfy  the  construction of the  present research,   we  will  make  the  following   agreements.   First  of all,  denote $s_0 =-4,$ $s_1 =-2,$ $s_2=0,$ $s_3 = 2,$ $s_4=4,$ $s_5= 6.$   Consider,  the  state  space of the process $S=\{s_0,s_1,s_2,s_3,s_4,s_5\},$ and the value $X_n  \in  S$  is the state of the process at time $n.$  The   Markov chain,   is   a random  process  which satisfy property    $P\{X_{n+1} = s_j|X_0, . . . , X_n\} = P\{X_{n+1} = s_j|X_n\}$ for   all $ s_i,s_ j \in  S$ and $n \ge 0,$ and, moreover,  $P\{X_{n+1} = s_j|X_n = s_i\} = p_{ij},$ where   $p_{ij}$  is the   transition  probability  that  the  chain jumps from state  $i$  to  state $j.$  It  is clear   that $\sum_{j=0}^{5}p_{ij}=1$ for all $i=0,...,5.$ The unpredictability of infinite realizations of the dynamics is approved by Theorem 2.2 \cite{Akhmet2022}.  

Next, we shall need the $\rho-$type piecewise constant unpredictable functions, which are defined through the Markov chain such that  $\rho(t) = X_n,$   if  $t \in [hn,h(n+1)).$ To  visualize the $\rho-$type functions in Figure \ref{f1} (a) the  graph  of  the  function $\rho(t) = X_n,$   if  $t \in [hn,h(n+1)),$ where $h=0.5,$  $0 \le t\le  100$ is drawn.   

On the basis of the $\rho-$type functions we introduce the $\sigma-$type piecewise constant unpredictable functions such that
\begin{equation}\label{sigma}
\sigma(t)=\Sigma(\rho(t)),
\end{equation}
where $\Sigma(s)$ is a continuous function, which satisfies the inverse Lipschitz condition. It can be shown that $\sigma-$type functions are discontinuous unpredictable \cite{Akhmetbook2021}. Figure \ref{f1} (b) depicts the graph of piecewise constant unpredictable function $\sigma(t)=\rho^2(t)+\rho(t).$

\begin{figure}[h]
	\begin{minipage}[h]{0.49\linewidth}
		\center{\includegraphics[width=0.8\linewidth]{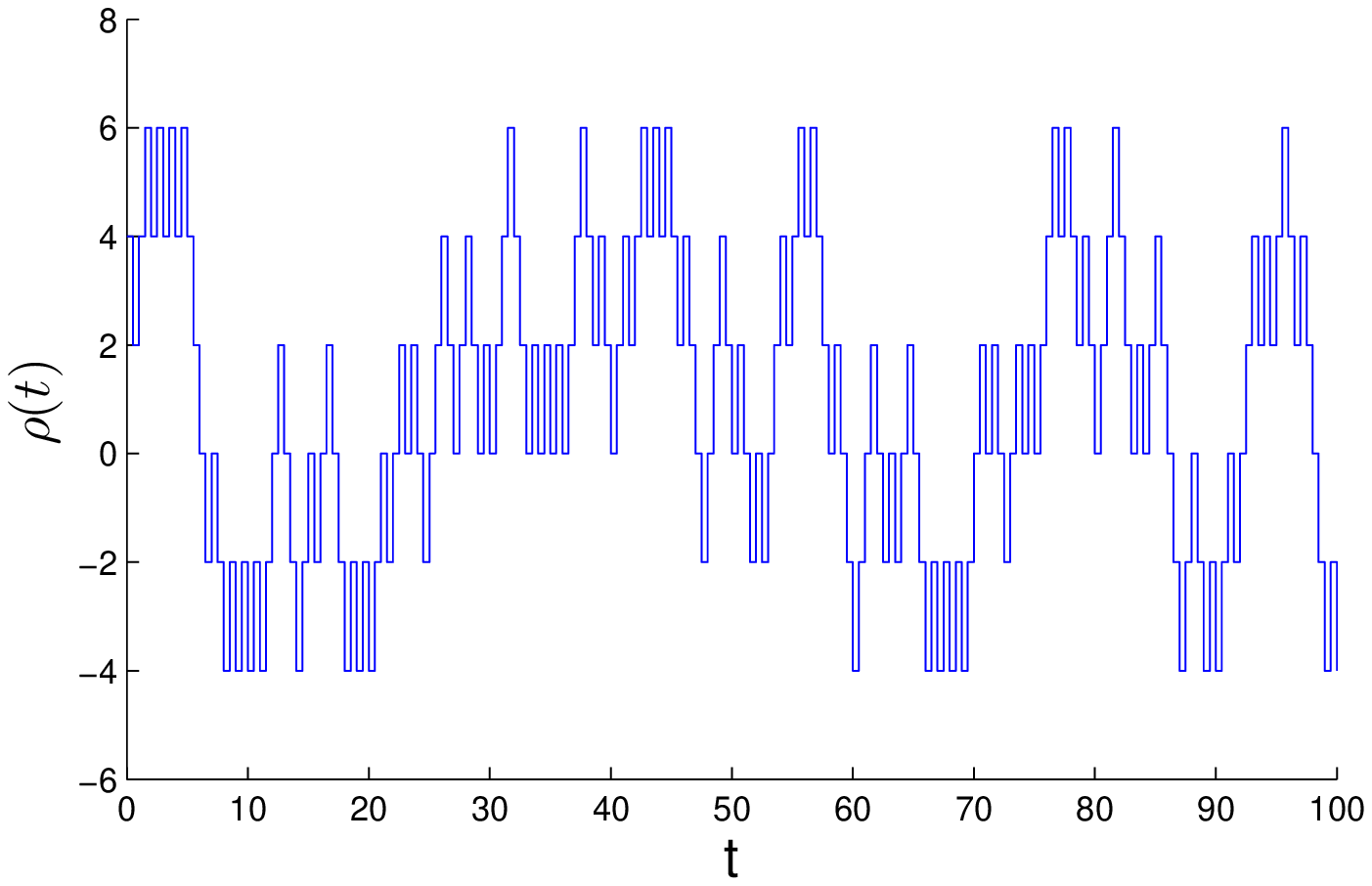} \\ a)}
	\end{minipage}
	\hfill
	\begin{minipage}[h]{0.49\linewidth}
		\center{\includegraphics[width=0.8\linewidth]{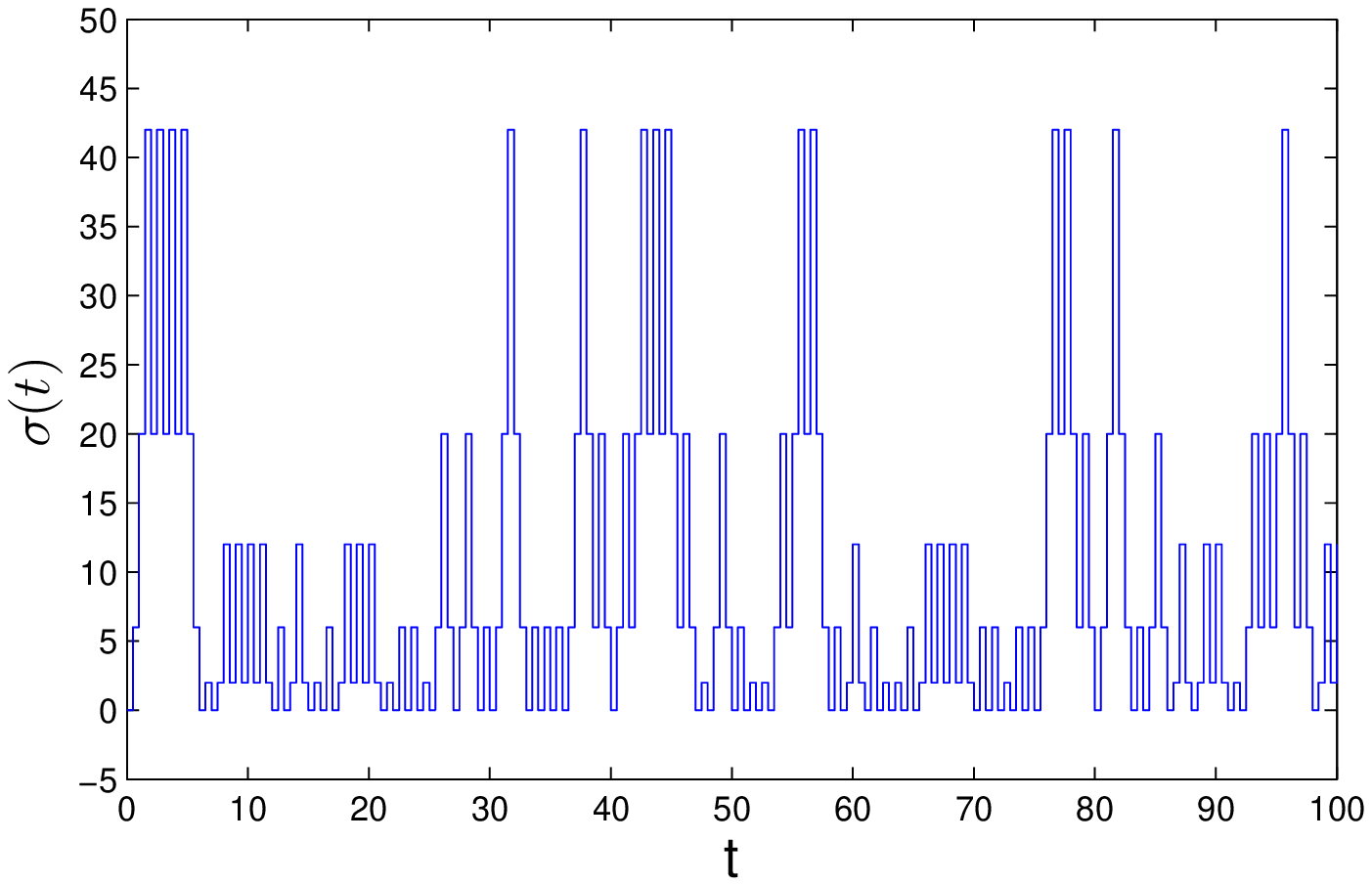} \\ b)}
	\end{minipage}
	\caption{The piecewise constant functions $\rho(t)$ and $\sigma(t).$ The  vertical  lines are  drawn  for better  visibility.}
	\label{f1}
\end{figure}


Now, let us define another type functions to finalize construction of continuous unpredictable functions through Markov process. Consider ordinary differential equation
\begin{equation}\label{Thetamarkov}
	x'(t)=\alpha x(t)+\sigma(t),
\end{equation}
 where $\alpha$ is a negative number.  The equation (\ref{Thetamarkov}) admits a unique exponentially stable unpredictable solution \cite{Akhmet2018}. We say that the solution of the equation (\ref{Thetamarkov}) is $\Theta-$type unpredictable function. It is impossible to specify the initial value of the solution, but applying the property of exponential stability one can consider any solution as arbitrary close. In Figure \ref{f2},  the graph of the solution, $x(t),$ $x(0)=0.6$ of equation (\ref{Thetamarkov}), where the parameter $\alpha$ is equal to -3, and $\sigma(t)=\rho^2(t)+\rho(t)$ is shown.  The solution exponentially approaches  an unpredictable function $\Theta(t).$ Thus, the algorithm for three types of unpredictable functions, which will be applied to build the Markovian coefficients has been finalized. Next, we shall apply it for each of the coefficients in the SDE (\ref{duffing}).

\begin{figure}[h]
	\center{\includegraphics[width=342pt,height=12pc]{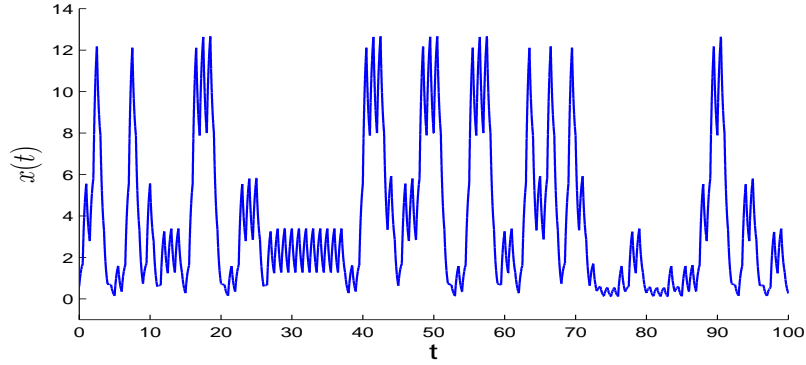}}
	\caption{The solution $x(t)$ of equation (\ref{Thetamarkov}) with initial value $x(0)=0.6$  exponentially approaches  the unpredictable Markovian function.  \label{f2}}
\end{figure}

Let us consider the following dissipative equations 
\begin{equation}\label{pmarkov}
	p^{'}_1(t)=\alpha_p p_1(t)+p(t), 
\end{equation}
\begin{equation}\label{qmarkov}
	q^{'}_1(t)=\alpha_q q_1(t)+q(t), 
\end{equation}
\begin{equation}\label{rmarkov}
	r^{'}_1(t)=\alpha_r r_1(t)+r(t), 
\end{equation}
\begin{equation}\label{fmarkov}
	F^{'}_1(t)=\alpha_F F_1(t)+f(t), 
\end{equation}
where $\alpha_p,$ $\alpha_q,$ $\alpha_r$ and $\alpha_F$ are negative real numbers, and $p(t),$ $q(t),$ $r(t)$ and $f(t)$ are unpredictable functions of $\sigma-$type. That is,  $p(t)=P(\rho(t)),$ $q(t)=Q(\rho(t)),$ $r(t)=R(\rho(t))$ and $f(t)=F(\rho(t)),$ where $P(s),$ $Q(s),$ $R(s),$ and $F(s),$ are continuous functions with inverse Lipschitz property and the function $\rho(t)$ is determined above. The exponentially stable and bounded solutions $p_1(t),$ $q_1(t),$ $r_1(t)$ and $F_1(t)$ of equations (\ref{pmarkov})-(\ref{fmarkov}) are $\Theta-$type functions. The functions are considered as Markovian components of the coefficients in the Duffing type equation (\ref{duffing}).

In this paper, we utilize Markov chains without memory for the coefficients, but it is clear that  one can consider chains with memories of arbitrary finite length in future studies.

\section{Main results}\label{sec3}

In the present section, under certain conditions, it is rigorously proved that an exponentially stable unpredictable solution takes place in the dynamics of the SDE with Markovian coefficients.

We will make use of the norm $\displaystyle \|v\|=\max(|v_1|,|v_2|),$ for a two-dimensional vector $v=(v_1,v_2),$ and corresponding norm for square matrices will be utilized. For SDE (\ref{duffing}) it is provided that a solution $x(t)$ and its derivative $x'(t)$ are bounded such that $\sup_{t \in \mathbb{R}} |x(t)|< H,$ $\sup_{t \in \mathbb{R}} |x'(t)|< H,$ where $H$ is a fixed positive number.

Assume that SDE (\ref{duffing}) satisfies the following conditions,
\begin{enumerate}[(C1)]
	\item the functions $p_0(t),$ $q_0(t),$ $r_0(t),$ and $F_0(t)$ are continuous periodic with common positive period $\omega$ such that $\lambda=\frac{2\pi}{\omega};$
	\item the Markovian components $p_1(t),$ $q_1(t),$ $r_1(t)$ $F_1(t)$ are of $\Theta-$type with common sequences of convergence $t_n,$ and divergence $s_n$ such that there exist positive numbers $\sigma,$ $\epsilon_0,$ which satisfy $|p_1(t+t_n)-p_1(t)|\geq \epsilon_{0},$ $|q_1(t+t_n)-q_1(t)|\geq \epsilon_{0},$ $|r_1(t+t_n)-r_1(t)|\geq \epsilon_{0},$ $|F_1(t+t_n)-F_1(t)|\geq \epsilon_{0},$ for all $t \in [s_n-\sigma; s_n+\sigma];$ 
	\item   $t_n\rightarrow 0 \ (mod \omega)$ as $n\rightarrow \infty;$
	\item  $s_n\rightarrow 0 \ (mod \frac{\omega}{2})$ as $n\rightarrow \infty.$
\end{enumerate}

The equation (\ref{duffing}) can be written as the  system
\begin{eqnarray} \label{system1}
	&&x_1'(t)=x_2(t), \nonumber \\
	&&x_2'(t)=-q_0(t)x_1(t)-p_0(t)x_2(t)-q_1(t)x_1(t)-p_1(t)x_2(t)-(r_0(t)+r_1(t))x^3_1(t)+(F_0(t)+F_1(t))cos(\lambda t).
\end{eqnarray}

Consider the homogeneous system, associated with (\ref{system1}),
\begin{eqnarray} \label{homo}
	&&x_1'(t)=x_2(t), \nonumber \\
	&&x_2'(t)=-q_0(t)x_1(t)-p_0(t)x_2(t).
\end{eqnarray}
 Let $X(t),$ $t \in \mathbb{R},$  is the fundamental matrix of system (\ref{homo}) such that $X(0)=I,$  and $I$  is the $2\times 2$ identical matrix.  Moreover,  $X(t,s)=X(t)X^{-1}(s)$ is the transition matrix of system (\ref{homo}) such that $X(t+\omega,s+\omega)=X(t,s)$ for all $t,s \in \mathbb{R}.$

The following assumption is needed, 
\begin{enumerate}[(C5)]
	\item the multipliers of system (\ref{homo})  in modulus are less than one.
\end{enumerate}

The last condition implies that there exist positive numbers $ K>1$ and $\mu,$ which satisfy 
\begin{eqnarray}\label{exponent}
	\|X(t,s)\|\leq Ke^{-\mu (t-s)}, 
\end{eqnarray}  
for $t\geq s$ \cite{Hartman2002}. 

For convenience, let introduce notations,
\begin{eqnarray*}
	q_0=\sup_{t \in \mathbb{R}} |q_0(t)|, p_0=\sup_{t \in \mathbb{R}} |p_0(t)|, p_1=\sup_{t \in \mathbb{R}} | p_1(t)|, q_1=\sup_{t \in \mathbb{R}}|q_1(t)|, r_0=\sup_{t \in \mathbb{R}}|r_0(t)|, r_1=\sup_{t \in \mathbb{R}}|r_1(t)|, F=\sup_{t \in \mathbb{R}}|F_0(t)+F_1(t)|.
\end{eqnarray*}

Throughout the paper, the following additional conditions are required,
\begin{enumerate}
	\item[(C6)] $\displaystyle  \frac{K}{\mu}(H(p_1+q_1)+(r_0+r_1) H^3 +F)<H;$
	\item[(C7)] $\displaystyle \frac{K}{\mu}(p_1+q_1+3(r_0+r_1) H^2)<1.$
\end{enumerate}

 We will consider the system (\ref{system1}) in the  matrix form
\begin{eqnarray}\label{system}
	y'=A(t)y+B(t)y+C(t,y)+G(t),
\end{eqnarray}
where $y(t)=column(y_1(t),y_2(t)),$
\begin{eqnarray*}
	A=\begin{pmatrix}
		0& 1\\
		-q_0(t)& -p_0(t)
	\end{pmatrix}, B(t)=\begin{pmatrix}
		0& 0\\
		-q_1(t)& -p_1(t)
	\end{pmatrix},	C(t,y)=\begin{pmatrix}
		0\\
		-(r_0(t)+r_1(t))y^3_1 
	\end{pmatrix},  G(t)=\begin{pmatrix}
		0\\
		(F_0(t)+F_1(t)) cos(\lambda t)
	\end{pmatrix}.
\end{eqnarray*}

Let us show the unpredictability of the function $G(t).$ Moreover, that the convergence and divergence sequences of the function are common with those for Markovian components. Fix a positive number $\epsilon,$  and a  bounded interval $I\subset \mathbb{R}.$ Duo to condition (C2), (C3), there exists a natural number $n_1$ such that 
\begin{eqnarray*}
	|\cos(\lambda(t+t_{n}))-\cos(\lambda t)|
	< \frac{\epsilon}{2F},
\end{eqnarray*}
 for all $t \in \mathbb{R}$ and $n>n_1.$ Besides it, there exists a natural number $n_2$ such that $$|F_0(t+t_n)+F_1(t+t_n)-F_0(t)-F_1(t)|< \frac{\epsilon}{2}$$ for all $t \in I$ and $n>n_2.$  Therefore, it is true that
\begin{eqnarray*}
	&&\|G(t+t_n)-G(t)\|=|(F_0(t+t_n)+F_1(t+t_n))\cos(\lambda (t+t_n)) -(F_0(t)+F_1(t))\cos(\lambda t)|\leq\\
	&& |F_0(t+t_n)+F_1(t+t_n)||\cos(\lambda (t+t_n))-\cos(\lambda t)|+|\cos(\lambda t)||F_0(t+t_n)+F_1(t+t_n)-F_0(t)-F_1(t)|\leq   F\frac{\epsilon}{2F}+\frac{\epsilon}{2}<\epsilon,
\end{eqnarray*}
for all $t \in I$ and $n>\max(n_1,n_2).$ On the other hand, there exist positive numbers $\epsilon_{0}, \sigma$ and sequence $u_n$ such that $|F_1(t+t_n)-F_1(t)|\geq \epsilon_{0}$ for each $t \in [s_n-\sigma,s_n+\sigma],$ $n \in \mathbb{N}.$ Moreover, for sufficiently large number $n$ one can attain that $|F_0(t+t_{n})-F_0(t)|<\frac{\epsilon_0}{8}$ and $|\cos(\lambda (t+t_n))-\cos(\lambda t)|<\frac{\epsilon_0}{8F},$ $t\in \mathbb{R}.$ Applying conditions (C3),(C4), we obtain that $\cos(\lambda(s_n+t_n))= 0 \ (mod \frac{\omega}{2})$ as $n \rightarrow \infty.$ Hence, due to the uniform continuity of the cosine function, there exists positive number $\sigma_1<\sigma$ such that $\min|\cos(\lambda(t+t_n))|>\frac{\epsilon_0}{2}$ for $t \in [s_{n}-\sigma_1,s_{n}+\sigma_1]$ and  $n \rightarrow \infty.$ This is why, we get that 
\begin{eqnarray*}
	&&\|G(t+t_n)-G(t)\|=|(F_0(t+t_n)+F_1(t+t_n))\cos(\lambda (t+t_n)) -(F_0(t)+F_1(t))\cos(\lambda t)|=\\
	&&\Big|(F_1(t+t_n)-F_1(t))\cos(\lambda (t+t_n))+(F_1(t)+F_0(t)(\cos(\lambda (t+t_n))-\cos(\lambda t))+(F_0(t+t_n)-F_0(t))\cos(\lambda (t+t_n))\Big|\\
	&&\geq|(F_1(t+t_n)-F_1(t))\cos(\lambda (t+t_n))|-|(F_1(t)+F_0(t))(\cos(\lambda (t+t_n))-\cos(\lambda t))|-|(F_0(t+t_n)-F_0(t))\cos(\lambda (t+t_n))|> \\
	&&\epsilon_{0}\min |\cos(\lambda(t+t_n))|-F\frac{\epsilon_0}{8F}-\frac{\epsilon_0}{8}>\frac{\epsilon_0}{4},
\end{eqnarray*}
for $t \in [s_{n}-\sigma_1,s_{n}+\sigma_1].$ Thus, the function $G(t)$ is unpredictable with sequences $t_n,$ $s_n,$ and positive numbers $\frac{\epsilon_0}{4},$ $\sigma_1$.

Condition (C5) implies that a bounded on the real axis function $z(t)$ is a solution of system (\ref{system}) if and only if it satisfies the equation  
\begin{equation}\label{integral}
	z(t)=\int_{-\infty}^t X(t,s)[B(s)z(s)+C(s,z(s))+G(s)]ds, \ t \in \mathbb{R}.
\end{equation}

Denote by $U$ the set of bounded and uniform continuous functions $v(t)=column(v_1(t), v_2(t)),$ with common convergence sequence $t_n$ such that $\left\|v(t) \right\|_0 < H,$ where $\displaystyle \left\|v(t) \right\|_0=\sup_{\mathbb{R}}\|v(t)\|.$  

Define on $U$ the operator $\varPhi$ as
\begin{equation}\label{operator}
	\varPhi v(t)=\int_{-\infty}^t X(t,s)(B(s)v(s)+C(s,v(s))+G(s))ds.
\end{equation}

\begin{lemma} \label{lemma2} The operator $\varPhi$ is invariant in $U.$ 
\end{lemma}

\begin{proof} Fix a function $v(t)$ that belongs to $U.$ We have that
\begin{eqnarray*}
	\|\varPhi v(t)\| \leq  \displaystyle \int_{-\infty}^t  \|X(t,s)\|(\|B(s)\|\|v(s)\|+\|C(s,v(s))\|+\|G(s)\|) ds \leq \frac{K}{\mu}((p_1+q_1)H+(r_0+r_1) H^3+F)
\end{eqnarray*}
for all $t\in\mathbb R.$ 
Therefore, by the condition  (C6) it is true that $\left\|\varPhi v \right\|_0 <H$.

Next, the method of included intervals \cite{AkhmetN2020,AkhmetZ2020} will be utilized to prove invariantness of Poisson stability in $U.$ Let us show that  $\|\varPhi v(t+t_n) -\varPhi v(t)\|\rightarrow 0$ on each bounded interval of $\mathbb R.$  Fix an arbitrary positive number $\epsilon$ and a closed interval $[a,b],$ $-\infty <a < b < \infty,$ of the real axis. Let us choose two numbers $c <a,$ and $\xi > 0$  satisfying 
\begin{eqnarray}\label{11}
	\displaystyle \frac{K}{\mu}(H(p_1+q_1)+(r_0+r_1)  H^3 +F)e^{-\mu(a-c)} <  \frac{\epsilon}{4},
\end{eqnarray} 
\begin{eqnarray}\label{12}
	\frac{K}{\mu} \xi (p_1+q_1+H+3(r_0+r_1)  H^2+H^3+1)[1-e^{-\mu(b-c)}] < \frac{\epsilon}{2}.
\end{eqnarray}
Conditions (C3), (C4) imply that for sufficiently  large $n$ the following inequalities are valid $\|B(t+t_n)-B(t)\|<\xi,$  $\|G(t+t_n)-G(t)\| < \xi,$ $| r_0(t+t_n)+r_1(t+t_n)-r_0(t)-r_1(t)|<\xi$ and  $\|v(t+t_n)-v(t)\|<\xi$ for $t\in[c,b].$ We obtain that 

\begin{eqnarray*}\label{Oper} 
&&\|\varPhi v(t+t_n) -\varPhi v(t)\|= \|\int_{-\infty}^{t} X(t,s)(B(s+t_n)v(s+t_n)+C(s+t_n,v(s+t_n))+G(s+t_n)) ds-\\
		&&\int_{-\infty}^{t} X(t,s)(B(s)v(s)+C(s,v(s))+G(s)) ds\| \leq \\
		&&\|\int_{-\infty}^{t}X(t,s)(B(s+t_n)v(s+t_n)-B(s)v(s)+C(s+t_n,v(s+t_n))-C(s,v(s))+\\
		&&G(s+t_n)-G(s)) ds \| \leq \int_{-\infty}^{c}\|X(t,s)\|\Big(\|B(s+t_n)v(s+t_n)-B(s)v(s)\|+\\
		&&\|C(s+t_n,v(s+t_n))-C(s,v(s))\|+\|G(s+t_n)-G(s)\|\Big) ds  +\\
		&&  \int_{c}^{t}\|X(t,s)\|\Big(\|B(s+t_n)(v(s+t_n)-v(s))\|+\|v(s)(B(s+t_n)-B(s))\|\Big)ds+\\
		&&\int_{c}^{t}\|X(t,s)\|\Big(\|C(s+t_n,v(s+t_n))-C(s+t_n,v(s))\|+ \nonumber \\
		&&\|C(s+t_n,v(s))-C(s,v(s))\|\Big)ds+
		\int_{c}^{t}\|X(t,s)\|\|G(s+t_n)-G(s)\| ds  \leq \\
		&&\frac{2K}{\mu}((p_1+q_1)H+(r_0+r_1) H^3+F)e^{-\mu(a-c)}+	\frac{K}{\mu} (\xi (p_1+q_1)+H\xi )[1-e^{-\mu(b-c)}]+\\
		&&	\frac{K}{\mu} (3\xi(r_0+r_1) H^2+\xi H^3)[1-e^{-\mu(b-c)}]+	\frac{K}{\mu}\xi[1-e^{-\mu(b-c)}],
\end{eqnarray*}
is correct for all  $t \in [a,b].$ From inequalities (\ref{11}) and (\ref{12}) it follows that
$\|\varPhi v(t+t_n) -\varPhi v(t)\| < \epsilon$
for $t\in[a,b].$ Therefore, the sequence $\varPhi v(t+t_n)$ uniformly converges to $\varPhi v(t) $  on each bounded interval of $\mathbb R.$ 

The function $\varPhi v(t)$ is a  uniformly continuous, since its derivative is a uniformly bounded on the real axis. Thus, the set $U$ is invariant for the operator $\varPhi$.
\end{proof}

\begin{theorem} \label{theorem} The SDE (\ref{duffing}) with Markovian coefficients admits a unique exponentially stable unpredictable solution provided that the conditions (C1)-(C7) are valid. Moreover, the divergence and convergence sequences of the output stochastic dynamics are common with those, $t_n$ and $s_n,$  of the stochastic components of the coefficients. 
\end{theorem}

\begin{proof} Let us prove completeness of the set $U.$ Consider  a Cauchy  sequence $\phi^k(t)$ in $U$,   which converges to a limit function $\phi(t)$ on $\mathbb{R}$. Fix a closed and bounded interval $I \subset \mathbb{R}.$ We get that
	\begin{eqnarray}\label{Poisson}
		\|\phi(t+t_n)-\phi(t)\|\le \|\phi(t+t_n)-\phi^k(t+t_n)\|+\|\phi^k(t+t_n)-\phi^k(t)\|+\|\phi^k(t)-\phi(t)\|. 	
	\end{eqnarray}
	One can choose sufficiently large $n$ and $k,$ such that each term on the right side of (\ref{Poisson}) is smaller than $\frac{\epsilon}{3}$ for an arbitrary  $\epsilon>0$ and $t\in I$. Thus, we conclude that the sequence $\phi(t+t_n)$ is uniformly converging to $\phi(t)$  on $I.$ That is, the set ${U}$   is complete.
	
	Next, we shall show that the operator $\varPhi : U\rightarrow  U$ is a contraction. For any $\varphi(t),$ $\psi(t) \in U,$ one can attain  that
	\begin{eqnarray*}
		&& \|\varPhi \varphi(t) -\varPhi \psi(t)\| 
		\leq \displaystyle \int_{-\infty}^t \|X(t,s)\|(\|B(s)\|\|\varphi(s)-\psi(s)\|+\|C(s,\varphi(s))-C(s,\psi(s))\|) ds \leq \\
		&&\frac{K}{\mu}\Big((p_1+q_1)\|\varphi(t)-\psi(t)\|_0+(r_0+r_1)(|\varphi^2_1(t)|+|\varphi_1(t)||\psi_1(t)|+|\psi^2_1(t)|)\|\varphi(t)-\psi(t)\|_0\Big)< \\
		&& \frac{K}{\mu}(p_1+q_1+3(r_0+r_1) H^2)\|\varphi(t)-\psi(t)\|_0. 
	\end{eqnarray*}
	Therefore,  the inequality $\left\|\varPhi \varphi - \varPhi \psi\right\|_0 <\displaystyle \frac{K}{\mu}(p_1+q_1+3(r_0+r_1) H^2) \left\|\varphi-\psi\right\|_0$ holds, and according to the condition (C7) the operator $\Pi:  U \to  U$ is a contraction. 
	
	By the contraction mapping theorem there exists  the unique fixed point, $z(t) \in U$  of the operator $\varPhi,$ which is the unique solution of SDE (\ref{duffing}). In what follows, we will show that the solution $z(t)$ is unpredictable. 
	
	Applying the relations
	\begin{eqnarray*}
		z(t)=z(s_n)+\int_{s_n}^{t}A(s)z(s)ds+\int_{s_n}^{t}B(s)z(s)ds+\int_{s_n}^{t}C(s,z(s))ds+\int_{s_n}^{t}G(s)ds
	\end{eqnarray*} 
	and
	\begin{eqnarray*}
		z(t+t_n)=z(s_n+t_n)+\int_{s_n}^{t}A(s+t_n)z(s+t_n)ds+\int_{s_n}^{t}B(s+t_n)z(s+t_n)ds+\int_{s_n}^{t}C(s+t_n,z(s+t_n))ds+\int_{s_n}^{t}G(s+t_n)ds
	\end{eqnarray*}
	we obtain that
		\begin{eqnarray*}
		&&z(t+t_n)-z(t)=z(s_n+t_n)-z(s_n)+\int_{s_n}^{t}(A(s+t_n)z(s+t_n)-A(s)z(s))ds+\int_{s_n}^{t}(B(s+t_n)z(s+t_n)-B(s)z(s))ds+\\
		&&\int_{s_n}^{t}(C(s+t_n,z(s+t_n))-C(s,z(s)))ds+\int_{s_n}^{t}(G(s+t_n)-G(s))ds.
	\end{eqnarray*}

Using conditions (C3), (C4) and uniform continuity of the entries of the matrix $A(t),$ periodic function $r_0(t)$ and solution $z(t),$ one can find a positive numbers $\sigma_2$ and integers $l,k, n_0$ such that the following inequalities are satisfied
	\begin{equation}\label{1}
		\sigma_2 <\sigma_1;
	\end{equation}
	\begin{equation}\label{2}
		\|A(t+t_n)-A(t)\|<\epsilon_0 (\frac{1}{l}+\frac{2}{k}), \quad t\in \mathbb{R}, n>n_0;
	\end{equation}
	\begin{equation}\label{3}
		|r_0(t+t_n)-r_0(t)|<\epsilon_0 (\frac{1}{l}+\frac{2}{k}), \quad t\in \mathbb{R}, n>n_0;
	\end{equation}
	\begin{eqnarray}\label{4}
		&&\frac{2l\sigma_2}{3}\Big(\epsilon_{0}\Big[\frac{1}{4}-(p_1+q_1+\max(q_0+p_0, 1)+H+3H^2(r_0+r_1)+H^3)(\frac{1}{l}+\frac{2}{k})\Big]-2(p_1+q_1)H-2H^3r_1\Big)\geq \epsilon_{0}; 
	\end{eqnarray}
	\begin{equation}\label{5}
		\|z(t+s)-z(t)\|<\epsilon_0 \min (\frac{1}{k},  \frac{1}{4l}),  \quad t\in \mathbb{R}, |s|<\sigma_2.
	\end{equation}
 Let the numbers  $\sigma_2, l$ and $k$ as well as numbers $n \in  \mathbb{N},$  be fixed. Consider the following two alternatives: (i)  $\|z(s_n+t_n)-z(s_n)\| <\epsilon_0/l;$ \quad (ii) $\|z(s_n+t_n)-z(s_n)\| \geq \epsilon_0/l.$  
	
	(i) Using (\ref{5}) one can show that 
	\begin{eqnarray}\label{omeg}
		\|z(t+t_n)-z(t_n)\|\leq
		\|z(t+t_n)-z(s_n+t_n)\|+ \|z(s_n+t_n)-z(s_n)\|+ \|z(s_n)-z(t)\|< \frac{\epsilon_0}{l}+\frac{\epsilon_0}{k}+\frac{\epsilon_0}{k} =\epsilon_0 (\frac{1}{l}+\frac{2}{k}), 
	\end{eqnarray}
	if $t \in [s_n, s_n+\sigma_2].$ 
	
	Therefore, the inequalities  (\ref{1})-(\ref{omeg}) imply that
	
	\begin{eqnarray*}
		&&\|z(t+t_n)-z(t)\|\geq \int_{s_n}^{t}\|G(s+t_n)-G(s)\|ds-\|z(s_n+t_n)-z(s_n)\|-\int_{s_n}^{t}\|B(s+t_n)z(s+t_n)-B(s)z(s)\|ds-\\
		&&\int_{s_n}^{t}\|A(s+t_n)z(s+t_n)-A(s)z(s)\|ds-\int_{s_n}^{t}\|C(s+t_n,z(s+t_n))-C(s,z(s))\|ds\geq\\
		&&\int_{s_n}^{t}\|G(s+t_n)-G(s)\|ds-\|z(s_n+t_n)-z(s_n)\|-\int_{s_n}^{t}\|B(s+t_n)-B(s)\|\|z(s+t_n)\|ds-\\
		&&\int_{s_n}^{t}\|B(s)\|\|z(s+t_n)-z(s)\|ds-\int_{s_n}^{t}\|A(s+t_n)-A(s)\|\|z(s+t_n)\|ds-\int_{s_n}^{t}\|A(s)\|\|z(s+t_n)-z(s)\|ds-\\
		&&\int_{s_n}^{t}|r_0(s+t_n)z^3_1(s+t_n)-r_0(s+t_n)z^3_1(s))|ds-\int_{s_n}^{t}|r_0(s+t_n)z^3_1(s)-r_0(s)z^3_1(s))|ds-\\
		&&\int_{s_n}^{t}|r_1(s+t_n)z^3_1(s+t_n))-r_1(s+t_n)z^3_1(s))|ds-\int_{s_n}^{t}|r_1(s+t_n)z^3_1(s)-r_1(s)z^3_1(s)|ds\geq\\
		&&\sigma_2\frac{\epsilon_{0}}{4}-\frac{\epsilon_0}{l}-2 \sigma_2(p_1+q_1)H-\sigma_2(p_1+q_1)\epsilon_0 (\frac{1}{l}+\frac{2}{k})-\sigma_2 \epsilon_0 (\frac{1}{l}+\frac{2}{k})H-\sigma_2 \max(q_0+p_0, 1)\epsilon_0 (\frac{1}{l}+\frac{2}{k})-\\
		&&3\sigma_2 r_0H^2 \epsilon_0 (\frac{1}{l}+\frac{2}{k})-\sigma_2 \epsilon_0 (\frac{1}{l}+\frac{2}{k})H^3-3\sigma_2 r_1H^2 \epsilon_0 (\frac{1}{l}+\frac{2}{k})-2 \sigma_2 r_1H^3>\frac{\epsilon_0}{2l}
	\end{eqnarray*}
	for $t \in[s_n,s_n+\sigma_2].$
	
	(ii) If $|z(t_n+s_n)-z(s_n)| \geq \epsilon_0/l$ it is not difficult to find that  (\ref{5}) implies
	\begin{eqnarray}
		\|z(t+t_n)-z(t)\|\geq  \|z(t_n+s_n)-z(s_n)\| - \|z(s_n)-z(t)\|- \|z(t+t_n)-z(t_n+s_n)\| \geq \frac{\epsilon_0}{k}-\frac{\epsilon_0}{l}-\frac{\epsilon_0}{4l}=\frac{\epsilon_0}{2l}, 
	\end{eqnarray}
	for  $t \in [s_n-\sigma_2, s_n+\sigma_2]$  and $n \in \mathbb N.$ Thus, it can be conclude that $z(t)$ is unpredictable solution with sequences $t_n,$ $s_n$ and positive numbers $\frac{\sigma_2}{2},$ $\frac{\epsilon_0}{2l}.$
	
Finally, let us discuss the exponential stability of the solution $z(t).$ It is true that 
\begin{eqnarray*}
	z(t)=X(t,t_0)z(t_0)+\int_{t_0}^t X(t,s)(B(s)z(s)+C(s,z(s))+G(s))ds.
\end{eqnarray*} 
Denote by $\bar{z}(t)$  another solution of SDE (\ref{duffing}) such that
\begin{eqnarray*}
	\bar{z}(t)=X(t,t_0)\bar{z}(t_0)+\int_{t_0}^t X(t,s)(B(s)\bar{z}(s)+C(s,\bar{z}(s))+G(s))ds.
\end{eqnarray*}
Making use of the relation
\begin{eqnarray*}
	\bar{z}(t)-z(t)=X(t,t_0)(\bar{z}(t_0)-z(t_0))+\int_{t_0}^t X(t,s)\Big(B(s)(\bar{z}(s)-z(s))+C(s,\bar{z}(s))-C(s,z(s))\Big)ds,
\end{eqnarray*}
one can obtain 

\begin{eqnarray}
	&&\|\bar{z}(t)-z(t)\|\leq \|X(t,t_0)\|\|\bar{z}(t_0)-z(t_0)\|+\int_{t_0}^t \|X(t,s)\|(\|B(s)\|\|\bar{z}(s)-z(s)\|+\|C(s,\bar{z}(s))-C(s,z(s))\|)ds \leq \nonumber \\
	&&Ke^{-\mu(t-t_0)}\|\bar{z}(t_0)-z(t_0)\|+\int_{t_0}^t Ke^{-\mu(t-s)}\Big((p_1+q_1)\|\bar{z}(s)-z(s)\|+(r_0+r_1) (|\bar{z}^2_1(s)|+|\bar{z}_1(s)|| z_1(t)|+| z^2_1(s)| )\|\bar{z}(s)-z(s)\|\Big)ds \nonumber \\
	&&\leq \frac{K}{\mu}(p_1+q_1+3(r_0+r_1) H^2)\|\bar{z}(t)-z(t)\|,  
\end{eqnarray}
for $t \in \mathbb{R}.$ With the aid of the Gronwall-Bellman Lemma, one can verify that 
\begin{eqnarray}\label{stable}
	\|\bar{z}(t)-z(t)\|\leq Ke^{(K(p_1+q_1+3(r_0+r_1) H^2)-\mu)(t-t_0)}\|\bar{z}(t_0)-z(t_0)\|,
\end{eqnarray}
for all $t\geq t_0,$ and condition  (C7) implies that the unpredictable solution, $z(t),$ is exponentially stable solution of SDE (\ref{duffing}). The theorem is proved. 
\end{proof}

The following section provides an example to confirm the theoretical results by using numerical simulations. It illustrates various unpredictable dynamics of the stochastic equation of Duffing type (\ref{duffing}) for different contributions of periodic and non-periodic components of coefficients.

\section{A numerical example and discussions}

Below, to visualize the exponentially stable unpredictable solution of $\Theta-$type and determine dynamics of Markov coefficients we shall apply solutions $\phi_p(t),$ $\phi_p(0)=0.5,$ $\phi_q(t),$ $\phi_q(0)=0.6,$ $\phi_r(t),$ $\phi_r(0)=0.4,$ and $\phi_F(t),$ $\phi_F(0)=0.3,$ of the dissipative equations (\ref{pmarkov})-(\ref{fmarkov}), where $\alpha_p=-5, \alpha_q=-3,\alpha_r=-2, \alpha_F=-4,$ and $p(t)=q(t)=r(t)=f(t)=\rho(t).$ The piecewise constant function $\rho(t),$ is constructed   by Markov chain with values over intervals $[hn, h(n+1)),$ $n \in \mathbb{N},$ and described in Section \ref{markfunction}. 
\begin{figure}[h]
	\begin{minipage}[h]{0.49\linewidth}
		\center{\includegraphics[width=0.8\linewidth]{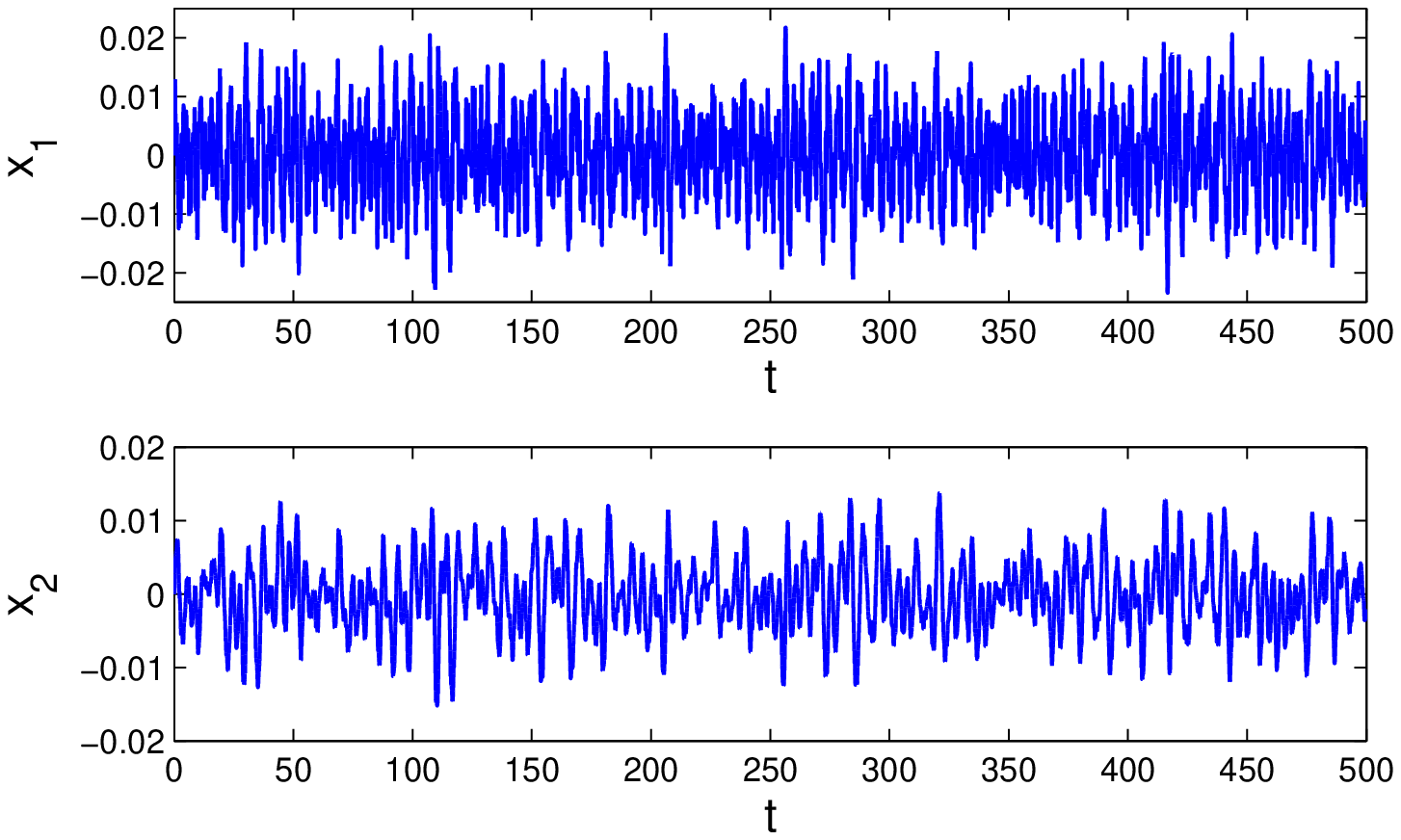} \\ a)}
	\end{minipage}
	\hfill
	\begin{minipage}[h]{0.49\linewidth}
		\center{\includegraphics[width=0.8\linewidth]{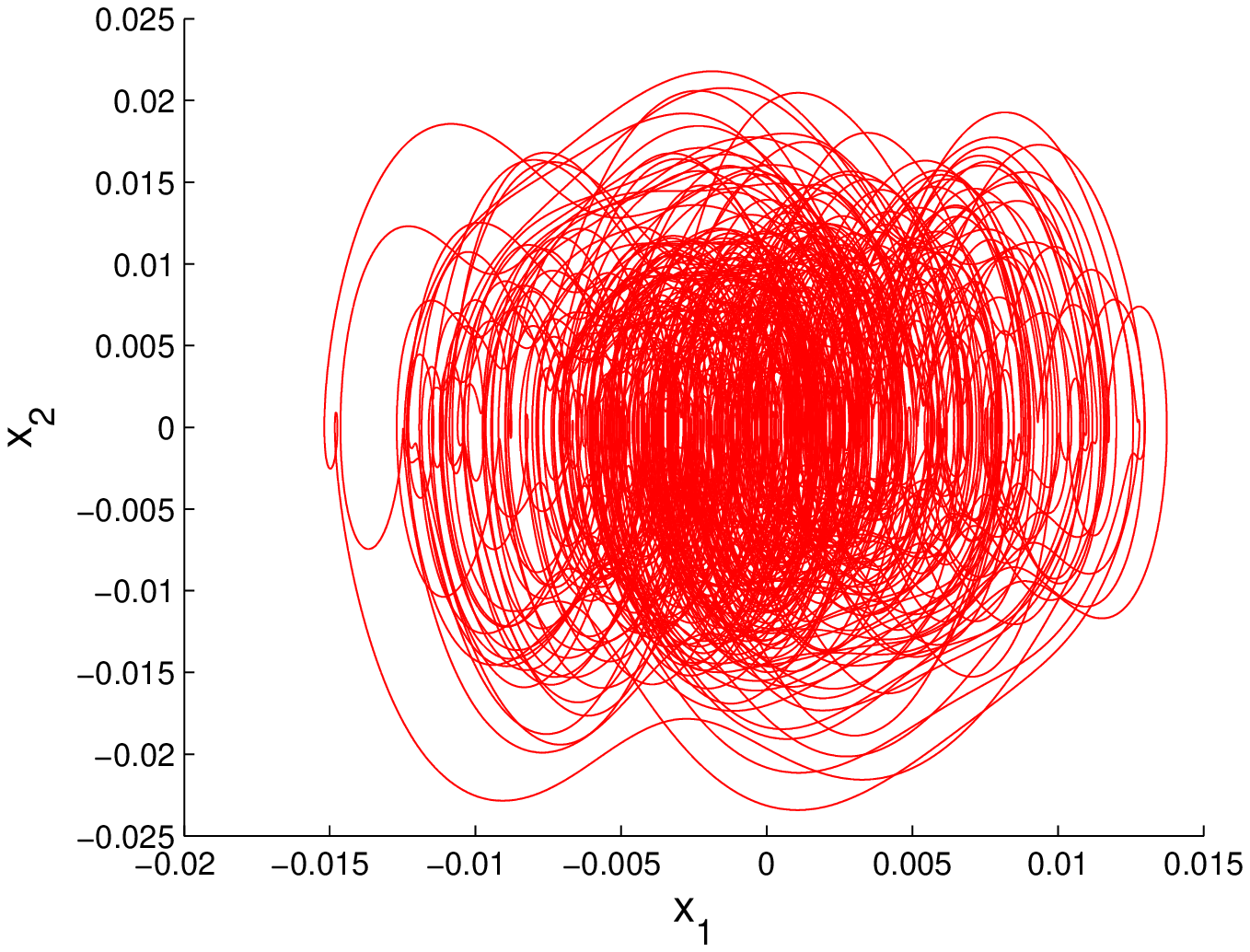} \\ b)}
	\end{minipage}
	\caption{The time series of the coordinates and trajectory of the solution $x(t)$ of equation (\ref{example}), with Markovian components obtained for $h=0.1\pi.$ The stochastic influence  is strong, since of the time step of the Markov chain is smaller than the period.}
	\label{sol1}
\end{figure}

Consider the following stochastic Duffing equation
\begin{eqnarray}\label{example}
	x''(t)+(p_0(t)+p_1(t))x'(t)+(q_0(t)+q_1(t))+(r_0(t)+r_1(t))x^3(t)=(F_0(t)+F_1(t))\cos(\lambda t),
\end{eqnarray}
where $\lambda=1,$ $p_0(t)=2-0.3\sin(4t),$ $q_0(t)=2-0.2\cos(2t),$ $r_0(t)=0.04\cos(4t),$ $F_0(t)=0.05\sin(8t),$ $p_1(t)=-0.4\phi_p(t),$ $q_1(t)=0.1\phi_q(t),$ $r_1(t)=0.02\phi_r(t)$ and $F_1(t)=0.02\phi_F(t).$ The periodic functions $p_0(t),$ $q_0(t),$ $r_0(t)$ and $F_0(t)$ with common period $\omega=2\pi.$ All conditions from (C1) to (C7) are hold with  $K=1.5,$ $\mu=2\pi,$ $p_0=2.3,$ $q_0=2.2,$ $r_0=0.04,$ $p_1=0.48,$ $q_1=0.2,$ $r_1=0.06,$ $F=0.03$ and $H=0.025.$ According to Theorem \ref{theorem}, the equation (\ref{example}) admits a unique exponentially stable unpredictable solution. In Figures \ref{sol1}-\ref{sol3}, the graphs of the coordinates and trajectories of solutions $x(t)$ for SDE (\ref{example}) with $h=0.2\pi, 2\pi, 8\pi,$ and initial values $x_1(0)=x_2(0)=0$ are shown. The solutions $x(t)$  exponentially approach the unpredictable solutions, $z(t),$ as time increases. 
\begin{figure}[h]
	\begin{minipage}[h]{0.49\linewidth}
		\center{\includegraphics[width=0.8\linewidth]{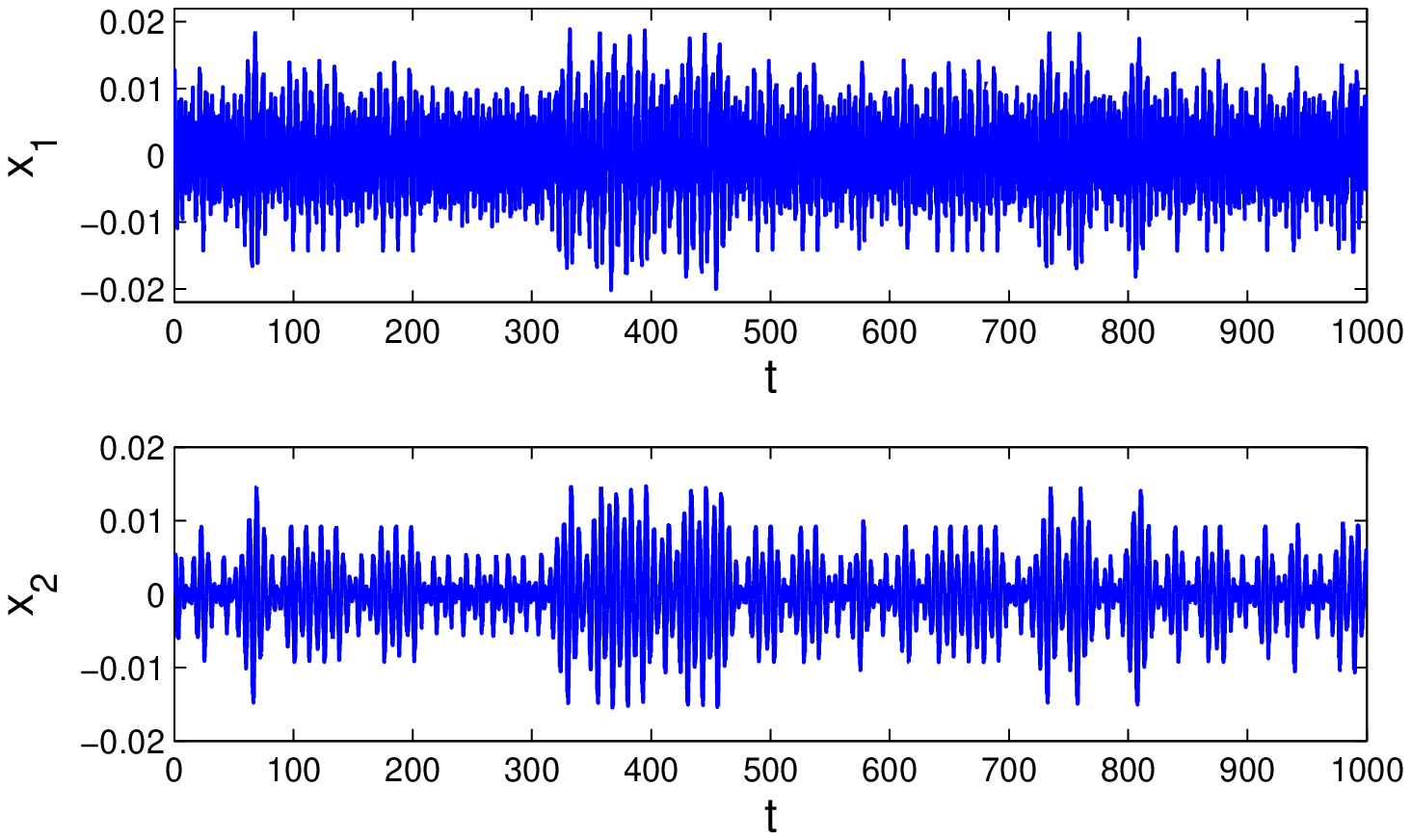} \\ a)}
	\end{minipage}
	\hfill
	\begin{minipage}[h]{0.49\linewidth}
		\center{\includegraphics[width=0.8\linewidth]{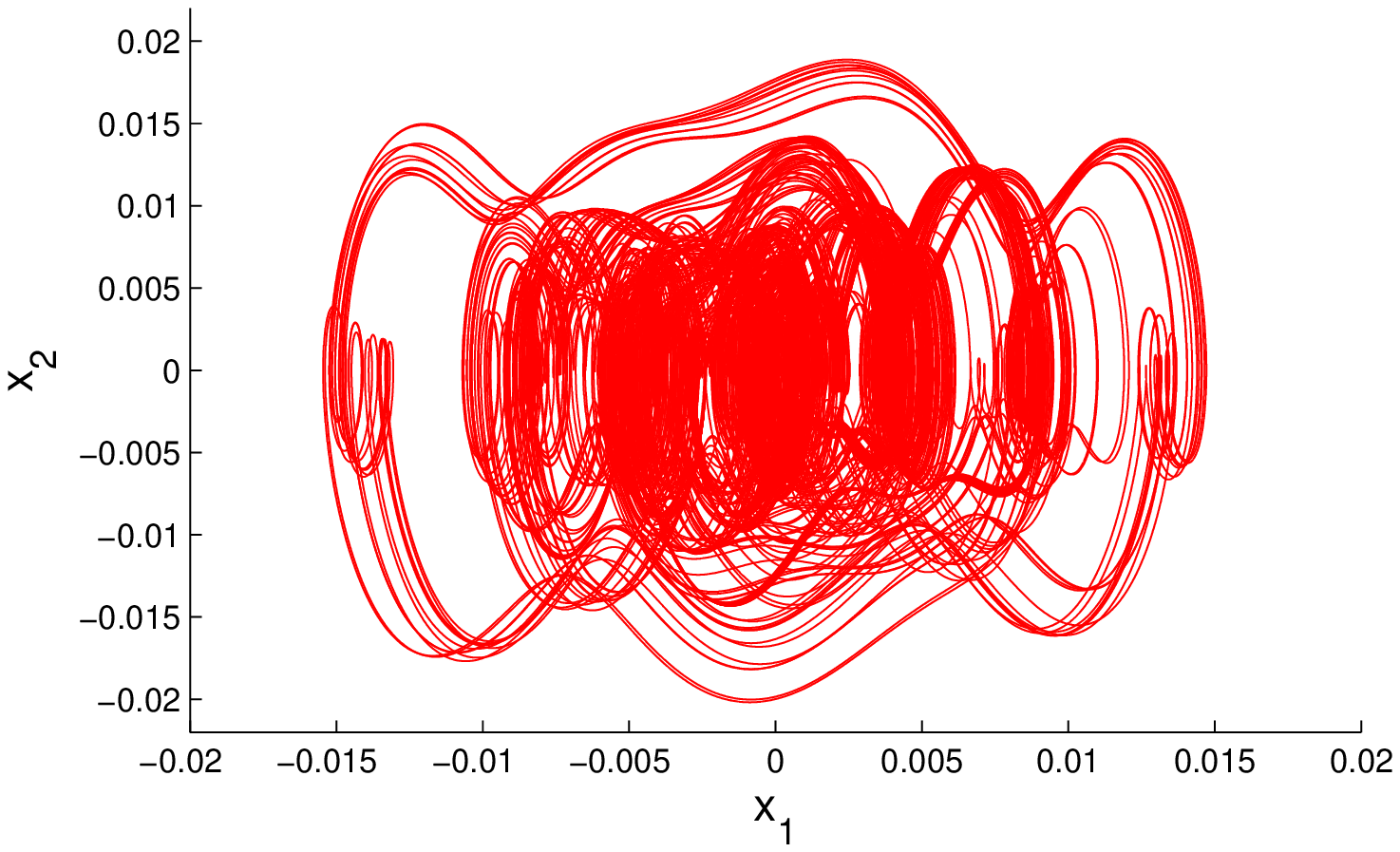} \\ b)}
	\end{minipage}
	\caption{The $x_1, x_2-$ coordinates and trajectory of the solution of (\ref{example}), with Markovian components obtained for $h=2\pi.$ That is, the value of time steps is equal to the period $2\pi$ and our simulations show that the periodicity still cannot be seen  clearly in this case.}
	\label{sol2}
\end{figure}

The Theorem \ref{theorem} can be interpreted as a result on response-driver synchronization \cite{Gonzalez2004} of the unpredictability  in	the stochastic system (\ref{pmarkov})-(\ref{fmarkov}) and the stochastic Duffing equation (\ref{duffing}). That is, the theorem claims, in particular, that the unpredictable solution $(p_1(t), q_1(t), r_1(t))$ of the system and the unpredictable solution, $z(t),$ admit common sequences of convergence and divergence. Delta synchronization of the unpredictability for gas discharge-semiconductor systems is considered in \cite{AkhmetK2022}. 
\begin{figure}[h]
	\begin{minipage}[h]{0.49\linewidth}
		\center{\includegraphics[width=0.8\linewidth]{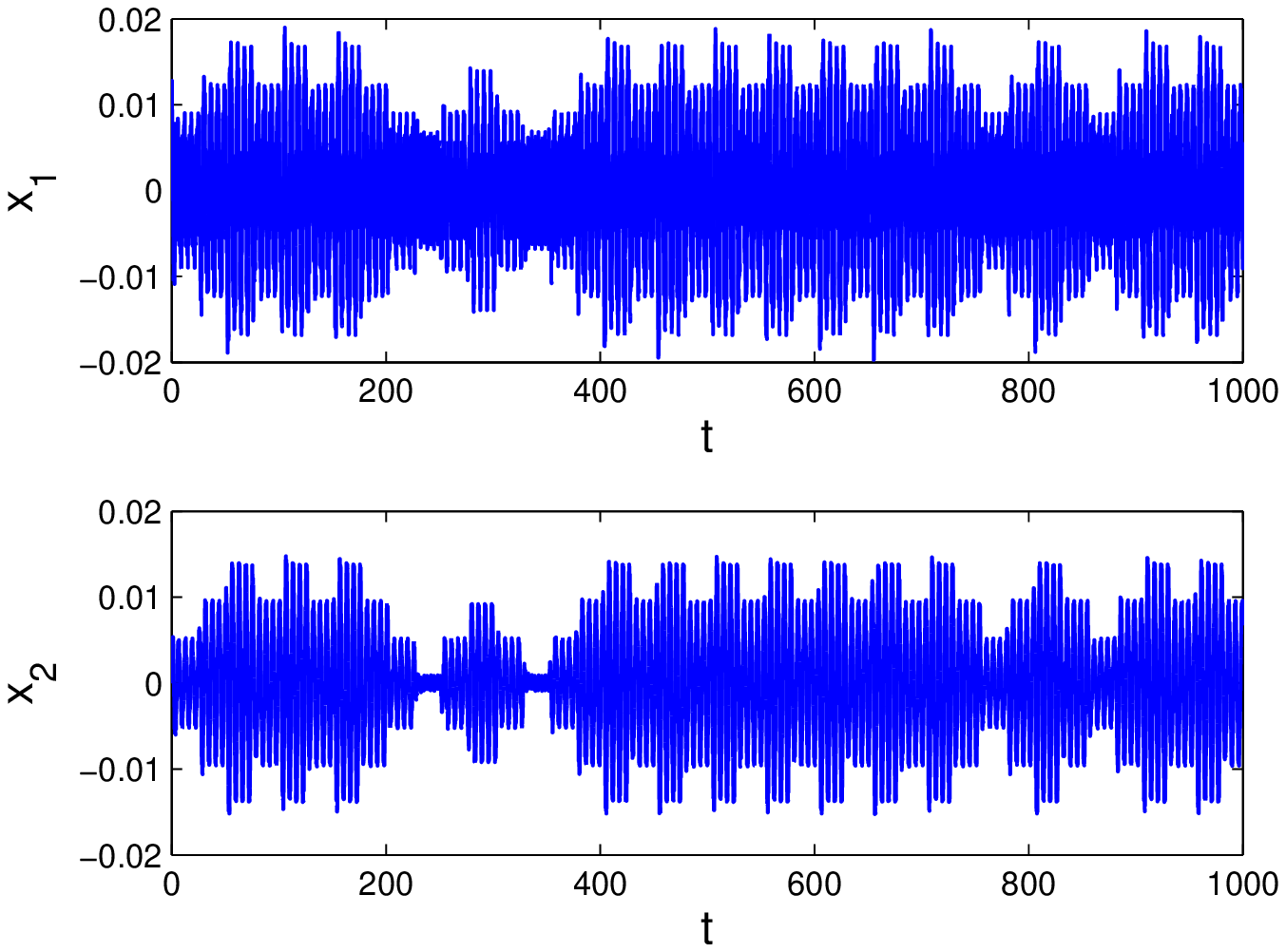} \\ a)}
	\end{minipage}
	\hfill
	\begin{minipage}[h]{0.49\linewidth}
		\center{\includegraphics[width=0.8\linewidth]{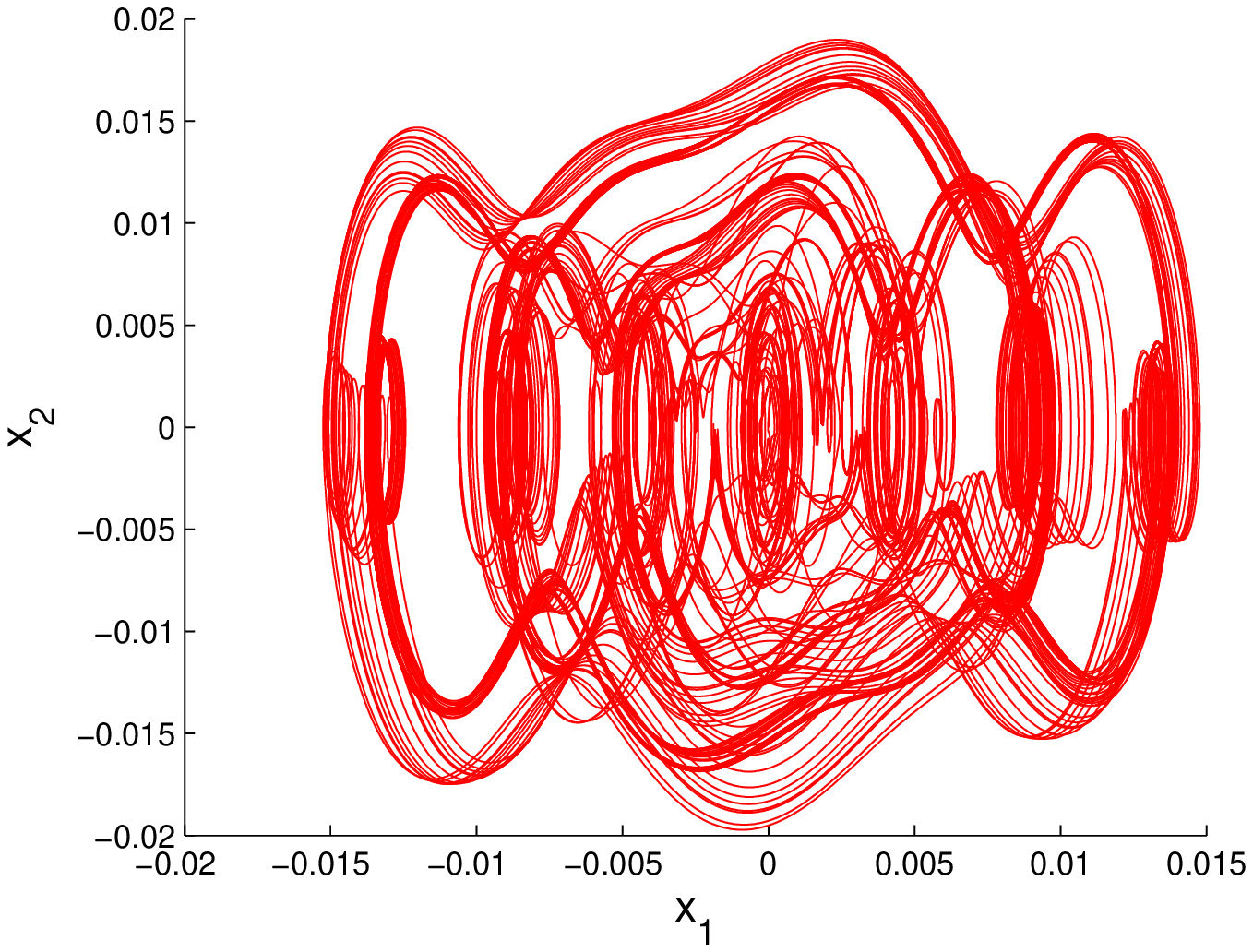} \\ b)}
	\end{minipage}
	\caption{The graphs of coordinates and trajectory of the solution $x(t)$ of equation (\ref{example}), where Markovian components obtained for $h=8\pi.$ One can see that several intervals of periodicity are placed in one step of  the constancy.}
	\label{sol3}
\end{figure}

 We consider the various simulations for the model, since they are with different steps $h$ of the Markov function $\rho(t).$  The choice makes qualitative difference in the stochastic dynamics. If the step is in the range $h \le 2\pi$ the behavior is strongly irregular, and there is no any indication of periodicity. For $h > 2\pi$ one can observe that   periodicity seen locally in time, and phenomenon of intermittency \cite{Pomeau1980} appears.  Thus, our results demonstrate not only quantitative asymptotic characteristics, but also possibility to learn reasons for different phenomena of chaos. Possibly, the simulations  may give lights on the origins of intermittency. Additionally, for $h \le 2\pi$ the effect of periodicity is seen in the "symmetry" of the phase portraits, which is reasoned  also by the finite values of the state space. For the values   of $h$ less than  $2\pi$, any symmetry  can not be seen, since the stochastic dynamics dominates significantly.





\section*{Acknowledgments}
M. Akhmet and A. Zhamanshin have  been supported by  2247-A National Leading Researchers Program of TUBITAK, Turkey, N 120C138. M. Tleubergenova  has been supported by the Science Committee of the Ministry of Education and Science of the Republic of Kazakhstan (grant No. AP14870835).

\subsection*{Author contributions}

Marat Akhmet:Conceptualization, formal analysis, investigation, methodology. Madina Tleubergenova:Formal analysis, investigation, supervision, validation. Akylbek Zhamanshin:Investigation, methodology, software.

\subsection*{Financial disclosure}

None reported.

\subsection*{Conflict of interest}

The authors declare no potential conflict of interests.

\nocite{*}
\bibliography{wileyNJD-AMA}%




\end{document}